\documentclass[american,aps,prl,reprint]{revtex4-1}
\usepackage[T1]{fontenc}
\usepackage[latin9]{inputenc}
\usepackage{color}
\usepackage{babel}
\usepackage{verbatim}
\usepackage{amsmath}
\usepackage{amsthm}
\usepackage{amssymb}
\usepackage{graphicx}
\usepackage[unicode=true,pdfusetitle,
 bookmarks=true,bookmarksnumbered=false,bookmarksopen=false,
 breaklinks=false,pdfborder={0 0 0},pdfborderstyle={},backref=false,colorlinks=true]
 {hyperref}
\hypersetup{
 allcolors=magenta}

\makeatletter
\theoremstyle{plain}
\newtheorem{thm}{\protect\theoremname}

\usepackage{babel}
\usepackage{babel}
\usepackage{babel}
\usepackage{times}
\usepackage{pxfonts}
\usepackage{braket}
\usepackage{colortbl}
\definecolor{myurlcolor}{rgb}{0,0,0.9}
\usepackage{bm}
\usepackage{bbm}

\newcommand{\EE}{\mathcal{E}}

\newcommand{\ran}{\rangle}
\newcommand{\lan}{\langle}

\providecommand{\theoremname}{Theorem}

\providecommand{\theoremname}{Theorem}

\addto\captionsbritish{\renewcommand{\theoremname}{Theorem}}
\addto\captionsenglish{\renewcommand{\theoremname}{Theorem}}
\providecommand{\theoremname}{Theorem}

\providecommand{\theoremname}{Theorem}

\providecommand{\theoremname}{Theorem}

\makeatother

\providecommand{\theoremname}{Theorem}

\begin{document}
\title{Entanglement negativity as a universal non-Markovianity witness}
\author{Jan Ko\l ody\'{n}ski}
\author{Swapan Rana}
\author{Alexander Streltsov}
\email{a.streltsov@cent.uw.edu.pl}

\affiliation{Centre for Quantum Optical Technologies, Centre of New Technologies,
University of Warsaw, Banacha 2c, 02-097 Warsaw, Poland}
\begin{abstract}
In order to engineer an open quantum system and its evolution, it
is essential to identify and control the memory effects. These are
formally attributed to the non-Markovianity of dynamics that manifests
itself by the evolution being indivisible in time, a property which
can be witnessed by a non-monotonic behavior of contractive functions
or correlation measures. We show that by monitoring directly the entanglement
behavior of a system in a tripartite setting it is possible to witness
all invertible non-Markovian dynamics, as well as all\emph{ }(also
non-invertible) qubit evolutions. This is achieved by using negativity,
a computable measure of entanglement, which in the usual bipartite
setting is not\emph{ }a universal non-Markovianity witness. We emphasize
further the importance of multipartite states by showing that non-Markovianity
cannot be faithfully witnessed by any contractive function of single
qubits. We support our statements by an explicit example of eternally
non-Markovian qubit dynamics, for which negativity can witness non-Markovianity
at arbitrary time scales. 
\end{abstract}
\maketitle
\textbf{\emph{Introduction.}} Describing effective dynamics of any
realistic quantum system that interacts with its environment inevitably
requires the theory of open quantum systems~\citep{Breuer,Lidar.A.2019}.
In recent years, a growing interest has been devoted to the determination
of dynamical properties that can be pinpointed when studying solely
the system evolution, in particular, distinguishing memory-less---\emph{Markovian---}dynamics
from ones that exhibit memory effects. Various ways have been proposed
on how to define the concept of memory or, more precisely, non-Markovianity
at the level of quantum evolutions, see~\citep{Rivas+2.RPP.2014,Breuer+3.RMP.2016,Vega+Alonso.RMP.2017,Li+2.PR.2018}
for detailed reviews on the topic. Although recently questioned \citep{Pollock2018,milz_cp_2019},
the most commonly adopted definition \citep{Rivas+2.PRL.2010,Chruscinski2014,bae_operational_2016}
is the natural generalisation of the Chapman-Kolgomorov equation,
which assures the time-divisibility of stochastic maps in case of
classical Markovian processes~\citep{Vacchini_2011}. In particular,
focusing on the family of quantum operations, i.e., completely positive
(CP) trace-preserving (TP) maps $\Lambda_{t}$ that represent the
system evolution from the initial time $t=0$ to each $t>0$, one
may verify their \emph{CP-divisibility }\citep{Wolf+Cirac.CMP.2008}
by inspecting whether at any intermediate time $0\le s\le t$ each
of them could be decomposed (concatenated) as 
\begin{equation}
\Lambda_{t}=V_{t,s}\circ\Lambda_{s}\label{cond:CP-div}
\end{equation}
with a valid dynamical (CPTP) map $V_{t,s}$.

Nevertheless, the above criterion is often weakened in order to construct
witnesses of non-Markovianity\emph{ }that despite not always being
able to certify the non-CP character of $V_{t,s}$ can have an operational
motivation. The most commonly used notion is the temporal behaviour
of distinguishability\emph{, }as measured by the trace distance $||\rho-\sigma||_{1}/2$
with the trace norm $||M||_{1}=\mathrm{Tr}\sqrt{M^{\dagger}M}$, between
a pair of evolving quantum states $\rho$ and $\sigma$ \citep{Breuer+2.PRL.2009}.
Its increase at a given time instance is then interpreted as a manifestation
of information backflow from the environment to the system \citep{Chruscinski+2.PRA.2011,buscemi_equivalence_2016}.

However, when dealing with invertible~\citep{Bylicka+2.PRL.2017}
or image non-increasing~\citep{Chruscinski+2.PRL.2018} dynamical
maps $\Lambda_{t}$, which describe almost all quantum evolutions,
the CP-divisibility criterion can be restated in terms of the information
backflow. By allowing for an ancilla of system dimension $d$, the
condition~(\ref{cond:CP-div}) becomes equivalent to the statement
\citep{Chruscinski+2.PRA.2011}:
\begin{equation}
\frac{\mathrm{d}}{\mathrm{d}t}\big\|\,\Lambda_{t}\otimes\openone_{d}[p_{1}\,\rho_{1}-p_{2}\,\rho_{2}]\,\big\|_{1}\leq0,\label{eq:CP-div_by_tr-norm}
\end{equation}
which must now be valid for all $t\ge0$, all bipartite system-ancilla
initial states $\rho_{1},\,\rho_{2}$ and all probabilities $p_{1}+p_{2}=1$
\footnote{In general, it is enough to consider only the right derivative in
Eq.~\eqref{eq:CP-div_by_tr-norm}.}. In this Letter, we will consider evolutions for which this equivalence
holds, what in fact includes also all qubit dynamics~\citep{Chakraborty+Chruscinski.A.2019}.
That is why, from now on we will refer to non-Markovianity as defined
by the violation of CP-divisibility.

Still, it has remained unknown whether such notion of non-Markovianity
can be faithfully verified by considering solely the evolution of
correlations, in particular, dynamics of the entanglement\emph{ }between
the system and some ancillae \citep{Rivas+2.PRL.2010}. This would
allow to certify non-Markovianity by preparing the system and ancillae
in an initial correlated state, in order to observe an increase of
some\emph{ }entanglement measure~\citep{VedralPhysRevLett.78.2275,Horodecki2009}
at a later time $t^{*}>0$, without need to consider ensembles of
initial states and distinguishability tasks~\citep{buscemi_equivalence_2016}.
Previous results suggest that traditional correlation quantifiers,
such as entanglement measures~\citep{Santis+4.PRA.2019,NetoPhysRevA.94.032105}
and mutual information~\citep{Luo+2.PRA.2012} fail to witness all
non-Markovian evolutions, while a recently proposed correlation measure~\citep{Santis+4.PRA.2019}
can witness ``almost all'' of them.

In this Letter, we show that negativity, a well known computable quantifier
of bipartite entanglement~\citep{ZyczkowskiPhysRevA.58.883,Vidal+Werner.PRA.2002},
can witness all non-Markovian qubit dynamics $\Lambda_{t}$ and all
invertible evolutions of arbitrary dimension. After discussing the
limitations in witnessing non-Markovianity in single-qubit systems,
we present the general construction for negativity as a universal
non-Markovianity witness. We provide an explicit example, witnessing
violations of CP-divisibility for eternally non-Markovian qubit evolutions~\citep{Hall2014}
at arbitrary time scales.

\medskip{}

\textbf{\emph{Witnessing non-Markovianity with contractive functions.}}
A general witness of non-Markovianity can be built from any contractive
function $f(\rho,\sigma)$ of two quantum states $\rho$ and $\sigma$,
where contractivity means that 
\begin{equation}
f(\Lambda[\rho],\Lambda[\sigma])\leq f(\rho,\sigma)
\end{equation}
for any quantum operation $\Lambda$. Important examples for contractive
functions are the trace distance $||\rho-\sigma||_{1}/2$, infidelity
$1-F(\rho,\sigma)$ with fidelity $F(\rho,\sigma)=||\sqrt{\rho}\sqrt{\sigma}||_{1}$,
and the quantum relative entropy $S(\rho||\sigma)=\mathrm{Tr}[\rho\log_{2}\rho]-\mathrm{Tr}[\rho\log_{2}\sigma]$.
Recently, a family of contractive functions, named quantum relative
Rényi entropy, has been introduced as~\citep{Muller-Lennert2013,Wilde2014}
\begin{equation}
D_{\alpha}^{\mathrm{q}}(\rho||\sigma)=\frac{1}{\alpha-1}\log_{2}\mathrm{Tr}\left[\left(\sigma^{\frac{1-\alpha}{2\alpha}}\rho\sigma^{\frac{1-\alpha}{2\alpha}}\right)^{\alpha}\right],\label{eq:QRR}
\end{equation}
with $\alpha\geq1/2$. In the limit $\alpha\rightarrow1$ the function
$D_{\alpha}^{\mathrm{q}}(\rho||\sigma)$ coincides with the relative
entropy $S(\rho||\sigma)$, and for $\alpha=1/2$ we obtain $D_{1/2}^{\mathrm{q}}(\rho||\sigma)=-2\log_{2}F(\rho,\sigma)$.

Noting that any contractive function is monotonically decreasing with
$t$ for any Markovian evolution, an increase of $f$ for some $t>0$
serves as a witness of non-Markovianity. It is now reasonable to ask
whether any non-Markovian evolution can be witnessed by some suitably
chosen contractive function. As we show in the Theorem~\ref{thm:Functions}
below, the answer to this question is negative for single-qubit systems.
An important type of evolutions in this context is given by Eq.~(\ref{cond:CP-div}),
where $V_{t,s}$ admits the decomposition 
\begin{equation}
V_{t,s}[\rho]=p\EE_{1}[\rho]+(1-p)\EE_{2}[\rho^{T}]\label{eq:V_decomp}
\end{equation}
with probabilities $p$ and CPTP-maps $\EE_{1}$ and $\EE_{2}$ which
can further depend on $t$ and $s$ with $s\le t$. Maps $V_{t,s}$
admitting Eq.~(\ref{eq:V_decomp}) are a subclass of positive maps
(P-maps) which are not necessarily CP, see Supplemental Material for
more details. Evolutions admitting decompositions with $V_{t,s}$
being P are generally called \emph{P-divisible}. An example of a non-Markovian
evolution admitting this form is presented below in Eq.~(\ref{eq:Pauli_map}).
We are now ready to present the first main result of this Letter. 
\begin{thm}
\label{thm:Functions}For any non-Markovian evolution $\Lambda_{t}=V_{t,s}\circ\Lambda_{s}$
with $V_{t,s}$ fulfilling Eq.~(\ref{eq:V_decomp}) it holds that:
\begin{equation}
\frac{\mathrm{d}}{\mathrm{d}t}f(\Lambda_{t}[\rho],\Lambda_{t}[\sigma])\leq0
\end{equation}
for any contractive function $f(\rho,\sigma)$ and any single-qubit
states $\rho$ and $\sigma$. 
\end{thm}

\begin{proof}
First, we will show that for any two single-qubit states $\rho$ and
$\sigma$ there exists a CPTP map $\Phi_{t,s}$ (that may in general
depend on both $\rho$ and $\sigma$) such that 
\begin{equation}
V_{t,s}[\rho]=\Phi_{t,s}[\rho],\,\,\,\,\,V_{t,s}[\sigma]=\Phi_{t,s}[\sigma].\label{eq:TransposeUnitary-1}
\end{equation}
This statement can be proven by considering the Bloch vectors $\boldsymbol{r}$
and $\boldsymbol{s}$ of the states $\rho$ and $\sigma$. The Bloch
vector $\tilde{\boldsymbol{r}}$ of the transposed state $\rho^{T}$
is related to $\boldsymbol{r}=(r_{x},r_{y},r_{z})$ via a reflection
on the $x$-$z$ plane, i.e., $\tilde{\boldsymbol{r}}=(r_{x},-r_{y},r_{z})$,
and similar for $\sigma$. In particular, this means that transposition
preserves the lengths of the two Bloch vectors and the angle between
them. This implies that for any two states $\rho$ and $\sigma$ there
exists a unitary rotation $U$ such that 
\begin{equation}
\rho^{T}=U\rho U^{\dagger},\,\,\,\,\,\sigma^{T}=U\sigma U^{\dagger}.\label{eq:TransposeUnitary-2}
\end{equation}
The CPTP map $\Phi_{t,s}$ fulfilling Eqs.~(\ref{eq:TransposeUnitary-1})
is thus given as 
\begin{equation}
\Phi_{t,s}[\rho]=p\EE_{1}\left[\rho\right]+(1-p)\EE_{2}\left[U\rho U^{\dagger}\right],
\end{equation}
where the unitary $U$ is chosen such that Eqs.~(\ref{eq:TransposeUnitary-2})
hold. Note that -- in general -- the unitary $U$ depends on the
two states $\rho$ and $\sigma$.

Combining the above arguments, we obtain the following for any contractive
function $f$ and any two single-qubit states $\rho$ and $\sigma$:
\begin{align}
f\left(\Lambda_{t}\left[\rho\right],\Lambda_{t}\left[\sigma\right]\right) & =f\left(V_{t,s}\circ\Lambda_{s}\left[\rho\right],V_{t,s}\circ\Lambda_{s}\left[\sigma\right]\right)\nonumber \\
 & =f\left(\Phi_{t,s}\circ\Lambda_{s}\left[\rho\right],\Phi_{t,s}\circ\Lambda_{s}\left[\sigma\right]\right)\nonumber \\
 & \leq f\left(\Lambda_{s}\left[\rho\right],\Lambda_{s}\left[\sigma\right]\right)
\end{align}
which proves that any contractive function is monotonically decreasing
with $t$. 
\end{proof}
While Theorem~\ref{thm:Functions} applies only to single-qubit systems,
this constraint can be lifted if one considers only specific functions,
namely the trace distance, the relative entropy, and the quantum relative
Rényi entropy $D_{\alpha}^{\mathrm{q}}(\rho||\sigma)$ for $\alpha>1$.
Noting that these functions are contractive under positive trace-preserving
maps~\citep{Muller+Reeb.AHP.2017}, it follows that they are monotonic
under non-Markovian evolutions which are P-divisible. We refer to
the Supplemental Material for more details.

A question which is left open in Theorem~\ref{thm:Functions} is
whether it is still possible to detect non-Markovianity via the behavior
of a contractive function $f$. Even if $f$ is monotonically decreasing
with $t$, its overall behavior might depend on whether the evolution
is Markovian or not. We answer this question in the Supplemental Material,
showing that the monotonic behavior of any contractive function can
be reproduced by Markovian dynamics.

\medskip{}

\textbf{\emph{Witnessing non-Markovianity with entanglement.}} The
results of the previous section tell us that to witness all non-Markovian
evolutions, our input state must be of higher dimension, possibly
a compound state of the system extended by ancillae, i.e., we need
to consider the evolution $\Lambda_{t}^{A}\otimes\openone^{B}$ acting
on a bipartite state $\rho=\rho^{AB}$. The behavior of any entanglement
measure $E^{A|B}$ of the final state 
\begin{equation}
\sigma_{t}=\Lambda_{t}^{A}\otimes\openone^{B}\left[\rho\right]
\end{equation}
then serves as a witness of non-Markovianty, as for any Markovian
evolution the entanglement must monotonically decrease~\citep{Rivas+2.PRL.2010}.
However, this approach is not suitable to create a universal witness
of non-Markovianity, as for any evolution $\Lambda_{t}$ which consists
of an entanglement breaking map at some finite time $t'$ followed
by an arbitrary non-Markovian evolution, the state $\sigma_{t}$ will
have zero entanglement for all $t\geq t'$~\citep{Santis+4.PRA.2019}.

Even if the evolution is not entanglement breaking, we can show that
certain entanglement quantifiers fail to detect non-Markovianity.
In the following, we quantify the amount of entanglement via negativity~\citep{ZyczkowskiPhysRevA.58.883,Vidal+Werner.PRA.2002}
\begin{equation}
E^{A|B}(\rho)=\frac{||\rho^{T_{B}}||_{1}-1}{2},\label{eq:Negativity}
\end{equation}
where $T_{B}$ denotes the partial transpose with respect to the subsystem
$B$. As is shown in the Supplemental Material, negativity is monotonic
under local positive maps of the form (\ref{eq:V_decomp}), i.e.,
\begin{equation}
P^{A}\otimes\openone^{B}[\rho]=p\EE_{1}^{A}\otimes\openone^{B}[\rho]+(1-p)\EE_{2}^{A}\otimes\openone^{B}\left[\rho^{T_{A}}\right],\label{eq:LocalPositive}
\end{equation}
for any bipartite state $\rho=\rho^{AB}$ and probability $p$ \footnote{Note that $P^{A}\otimes\openone^{B}[\rho]$ might not be positive,
we thus extend the definition of negativity in Eq.~(\ref{eq:Negativity})
to non-positive Hermitian operators with unit trace.}. This implies that negativity is monotonically decreasing for any
local evolution $\Lambda_{t}^{A}=V_{t,s}^{A}\circ\Lambda_{s}^{A}$
with $V_{t,s}$ being of the form~(\ref{eq:V_decomp}). An example
for a non-Markovian evolution admitting this form will be given in
Eq.~(\ref{eq:Pauli_map}). As we further show in the Supplemental
Material, negativity cannot be used to witness non-Markovianity if
$E^{A|B}(\Lambda_{t}^{A}\otimes\openone^{B}[\rho])$ is monotonically
decreasing with $t$, as a decreasing behavior can always be reproduced
by Markovian dynamics. From this, we conclude that negativity $E^{A|B}$
fails to witness some non-Markovian evolutions on subsystem $A$ even
if they are not entanglement breaking~\footnote{Note that negativity in general fails to detect non-Markovianity for
evolutions $\Lambda_{t}$ which consists of an NPT breaking map for
some $t'>0$, followed by an arbitrary evolution for $t>t'$.}.

In the light of these results, it is tempting to conclude that negativity
is not suitable for construction of a universal non-Markovianity witness.
Quite surprisingly, the situation changes completely by adding an
extra particle $C$, and considering the negativity $E^{AB|C}$ of
the state 
\begin{equation}
\tau_{t}^{ABC}=\Lambda_{t}^{A}\otimes\openone^{BC}\left[\rho^{ABC}\right],
\end{equation}
where $\rho^{ABC}$ is a suitably chosen initial state. In fact, taking
additional ancilla systems into account has proven to be useful for
relating different notions of non-Markovianity, see Eq.~(\ref{eq:CP-div_by_tr-norm}).
The following theorem shows that in a tripartite setting negativity
is a universal non-Markovianity witness for all invertible evolutions
and for all dynamics of a single qubit.
\begin{thm}
\label{thm:neg_is_awesome}For any invertible non-Markovian evolution
$\Lambda_{t}$ there exists a quantum state $\rho^{ABC}$ such that
\emph{
\begin{equation}
\frac{\mathrm{d}}{\mathrm{d}t}E^{AB|C}\left(\Lambda_{t}^{A}\otimes\openone^{BC}\left[\rho^{ABC}\right]\right)>0
\end{equation}
}for some $t>0$. For single-qubit evolutions $\Lambda_{t}$ the statement
also holds for non-invertible dynamics.
\end{thm}

\begin{proof}
We introduce the following state 
\begin{equation}
\rho^{ABC}=p_{1}\rho_{1}^{AB_{1}}\otimes\ket{\Psi^{+}}\!\bra{\Psi^{+}}^{B_{2}C}+p_{2}\rho_{2}^{AB_{1}}\otimes\ket{\Psi^{-}}\!\bra{\Psi^{-}}^{B_{2}C},\label{eq:InitialState}
\end{equation}
where $B_{1}$ and $B_{2}$ are subsystems of $B=B_{1}B_{2}$, $|\Psi^{\pm}\ran=(|01\ran\pm|10\ran)/\sqrt{2}$
are maximally entangled states, and the states $\rho_{i}$ and probabilities
$p_{i}$ will be specified in more detail below. If now an evolution
$\Lambda_{t}^{A}$ acts on the state $\rho^{ABC}$, the time-evolved
state takes the form 
\begin{align}
\tau_{t}^{ABC} & =p_{1}\Lambda_{t}^{A}\left[\rho_{1}^{AB_{1}}\right]\otimes\ket{\Psi^{+}}\!\bra{\Psi^{+}}^{B_{2}C}\nonumber \\
 & +p_{2}\Lambda_{t}^{A}\left[\rho_{2}^{AB_{1}}\right]\otimes\ket{\Psi^{-}}\!\bra{\Psi^{-}}^{B_{2}C}.
\end{align}
To evaluate the negativity in $AB|C$ cut we notice that the partial
transposition with respect to $C$ is given by 
\begin{align}
\tau_{t}^{T_{C}} & =\frac{1}{2}\Lambda_{t}^{A}\left[p_{1}\rho_{1}^{AB_{1}}+p_{2}\rho_{2}^{AB_{1}}\right]\otimes\left(|01\ran\lan01|^{B_{2}C}+|10\ran\lan10|^{B_{2}C}\right)\nonumber \\
 & +\frac{1}{2}\Lambda_{t}^{A}\left[p_{1}\rho_{1}^{AB_{1}}-p_{2}\rho_{2}^{AB_{1}}\right]\otimes\left(|\Phi^{+}\ran\lan\Phi^{+}|^{B_{2}C}-|\Phi^{-}\ran\lan\Phi^{-}|^{B_{2}C}\right)\nonumber \\
\end{align}
with $|\Phi^{\pm}\ran=(|00\ran\pm|11\ran)/\sqrt{2}$. Since the states
$\ket{\Phi^{\pm}}$ are orthogonal to $\ket{01}$ and $\ket{10}$,
the trace norm of $\tau_{t}^{T_{C}}$ can be evaluated as 
\begin{equation}
\left\Vert \tau_{t}^{T_{C}}\right\Vert _{1}=1+\left\Vert \Lambda_{t}^{A}\left[p_{1}\rho_{1}^{AB_{1}}-p_{2}\rho_{2}^{AB_{1}}\right]\right\Vert _{1},
\end{equation}
where we used the fact that $\mu:=p_{1}\Lambda_{t}^{A}[\rho_{1}^{AB_{1}}]+p_{2}\Lambda_{t}^{A}[\rho_{2}^{AB_{1}}]$
is a valid quantum state, and thus $||\mu||_{1}=1$. The negativity
of $\tau_{t}^{ABC}$ is thus given as 
\begin{equation}
E^{AB|C}\left(\tau_{t}^{ABC}\right)=\frac{1}{2}\left\Vert \Lambda_{t}^{A}\left[p_{1}\rho_{1}^{AB_{1}}-p_{2}\rho_{2}^{AB_{1}}\right]\right\Vert _{1}.\label{eq:neg_witness}
\end{equation}
To complete the proof of the theorem, recall that for any invertible
evolution there exists states $\rho_{i}^{AB_{1}}$ and probabilities
$p_{i}$ such that Eq.~(\ref{eq:CP-div_by_tr-norm}) is violated
if the evolution is non-Markovian~\citep{Chruscinski+2.PRA.2011,Bylicka+2.PRL.2017}.
The same is true for all (also non-invertible) single-qubit dynamics~\citep{Chakraborty+Chruscinski.A.2019}.
\end{proof}
Few remarks regarding Theorem~\ref{thm:neg_is_awesome} are in place.
First, we note that invertible dynamics constitute the generic case
of quantum evolutions, as non-invertible evolutions have zero measure
in the space of all quantum evolutions~\citep{Ott2005,Santis+4.PRA.2019}.
Moreover, the statement of Theorem~\ref{thm:neg_is_awesome} can
be lifted to include also dynamics which are image non-increasing,
by applying the same arguments~\citep{Chruscinski+2.PRL.2018}. We
further notice that negativity is a faithful entanglement quantifier
in the setting considered here, and the states in Eq.~(\ref{eq:InitialState})
are never bound entangled, see Supplemental Material for more details.

\medskip{}

\textbf{\emph{Applications. }}We apply the results presented above
to qubit eternally non-Markovian (ENM) dynamics~\citep{Hall2014},
an evolution exhibiting non-Markovianity at any $t>0$, even at arbitrarily
small and large timescales. Such a model falls into well-studied categories
of random-unitary \citep{Chruscinski2013} and phase-covariant \citep{Smirne2016}
qubit commutative evolutions. Yet, it constitutes an important example
with its non-Markovian features being hard to witness~\citep{megier_eternal_2017,chen_hierarchy_2015}.
In general, a random-unitary qubit dynamics is described by a time-dependent
master equation:
\begin{equation}
\frac{\mathrm{d}\rho(t)}{\mathrm{d}t}=\sum_{i=1}^{3}\gamma_{i}(t)\left\{ \sigma_{i}\rho(t)\sigma_{i}-\rho(t)\right\} ,\label{eq:RU_master_eq}
\end{equation}
which upon integration yields a dynamical map corresponding to a qubit
Pauli channel, i.e.: 
\begin{equation}
\Lambda_{t}\left[\rho\right]=\sum_{\mu=0}^{3}p_{\mu}(t)\sigma_{\mu}\rho\sigma_{\mu},\label{eq:Pauli_map}
\end{equation}
where the mixing probabilities $p_{\mu}(t)$, and their time-dependence,
can be explicitly expressed as a function of $\gamma_{i}(t)$ \citep{Chruscinski2013}.
For any such evolution the CP-divisibility condition (\ref{cond:CP-div})
is equivalent to the statement that for all $t>0$ all the decay rates
are non-negative, $\gamma_{i}(t)\ge0$, while the P-divisibility criterion
corresponds to a weaker requirement that at all times $t>0$ each
pair ($i\ne j$) of decay parameters satisfies $\gamma_{i}(t)+\gamma_{j}(t)\ge0$
\citep{Chruscinski2013}.

The ENM model introduced in Ref.~\citep{Hall2014} corresponds then
to the choice: 
\begin{equation}
\gamma_{1}=\gamma_{2}=\alpha\frac{c}{2},\quad\gamma_{3}(t)=-\alpha\frac{c}{2}\tanh(ct)\label{eq:ENM_model}
\end{equation}
with $\alpha\ge1$ and $c>0$. Crucially, ENM dynamics exhibits non-Markovianity
at all times, as $\gamma_{3}(t)<0$ for all $t>0$. In contrast, it
is always P-divisible due to $\gamma_{\ell}+\gamma_{3}(t)=\alpha\frac{c}{2}(1-\tanh(ct))\ge0$
for $\ell\in\{1,2\}$ and any $t\ge0$ \citep{Benatti2017,Kolodynski2018}.
Still, the resulting CP-map~\eqref{eq:Pauli_map} is\emph{ }invertible,
i.e., for every $t\ge0$ one can find a linear map $\Lambda_{t}^{-1}$
such that $\Lambda_{t}^{-1}\circ\Lambda_{t}=\openone$. As a result,
one can unambiguously define $V_{t,s}=\Lambda_{t}\circ\Lambda_{s}^{-1}$
in \eqref{cond:CP-div} and explicitly compute its Choi-Jamio\l kowski
(CJ) matrix, $\Omega_{V_{t,s}}:=2\,V_{t,s}\otimes\openone\left[\ket{\Phi^{+}}\!\bra{\Phi^{+}}\right]$,
associated with it: 
\begin{equation}
\Omega_{V_{t,s}}=\frac{1}{2}\left(\begin{array}{cccc}
1+\lambda_{t-s}^{2\alpha} & 0 & 0 & 2\Gamma_{t,s}^{\alpha}\\
0 & 1-\lambda_{t-s}^{2\alpha} & 0 & 0\\
0 & 0 & 1-\lambda_{t-s}^{2\alpha} & 0\\
2\Gamma_{t,s}^{\alpha} & 0 & 0 & 1+\lambda_{t-s}^{2\alpha}
\end{array}\right),\label{eq:ENM_CJmat}
\end{equation}
where $\lambda_{\tau}=\mathrm{e}^{-c\tau}$ and $\Gamma_{t,s}=\lambda_{t-s}\cosh(ct)\text{sech}(cs)$.
It may be explicitly verified that $\Omega_{V_{t,s}}$ is non-positive
for any $0<s<t$, confirming the ``eternal non-Markovianity'' of
dynamics, unless $s=0$ for which $\Omega_{V_{t,0}}=\Omega_{\Lambda_{t}}\ge0$
assures the physicality of the overall evolution.

In the Supplemental Material, we explicitly show that the CJ-matrix~\eqref{eq:ENM_CJmat}
admits a convex decomposition: 
\begin{equation}
\Omega_{V_{t,s}}=p_{1}P_{\Phi_{+}}+p_{2}P_{\Phi_{-}}+(1-p_{1}-p_{2})P_{\Psi_{+}}^{T_{B}}\label{eq:ENM_decomp}
\end{equation}
with probabilities $p_{1}=\frac{1}{2}\left(\lambda_{t-s}^{2\alpha}+\Gamma_{t,s}^{\alpha}\right)$
and $p_{2}=\frac{1}{2}\left(1-\Gamma_{t,s}^{\alpha}\right)$, and
$P_{\psi}=2\ket{\psi}\!\bra{\psi}$. Hence, it follows (see Supplemental
Material for a general discussion) that the decomposition \eqref{eq:ENM_decomp}
of the CJ-matrix assures the map $V_{t,s}$ for the ENM dynamics to
admit a decomposition \eqref{eq:V_decomp}. As a direct consequence,
Theorem~\ref{thm:Functions} applies to the ENM dynamics, implying
that no contractive function $f(\rho,\sigma)$ evaluated on single-qubit
states $\rho$ and $\sigma$ will be able to witness non-Markovianity
of the ENM model. Moreover, as Eq.~\eqref{eq:V_decomp} naturally
generalizes to Eq.~\eqref{eq:LocalPositive}, it becomes evident
that negativity cannot be used in the usual bipartite setting $E^{A|B}(\Lambda_{t}^{A}\otimes\openone^{B}[\rho])$
to witness the non-Markovianity of the ENM evolution.

\begin{figure}
\includegraphics[width=1\columnwidth]{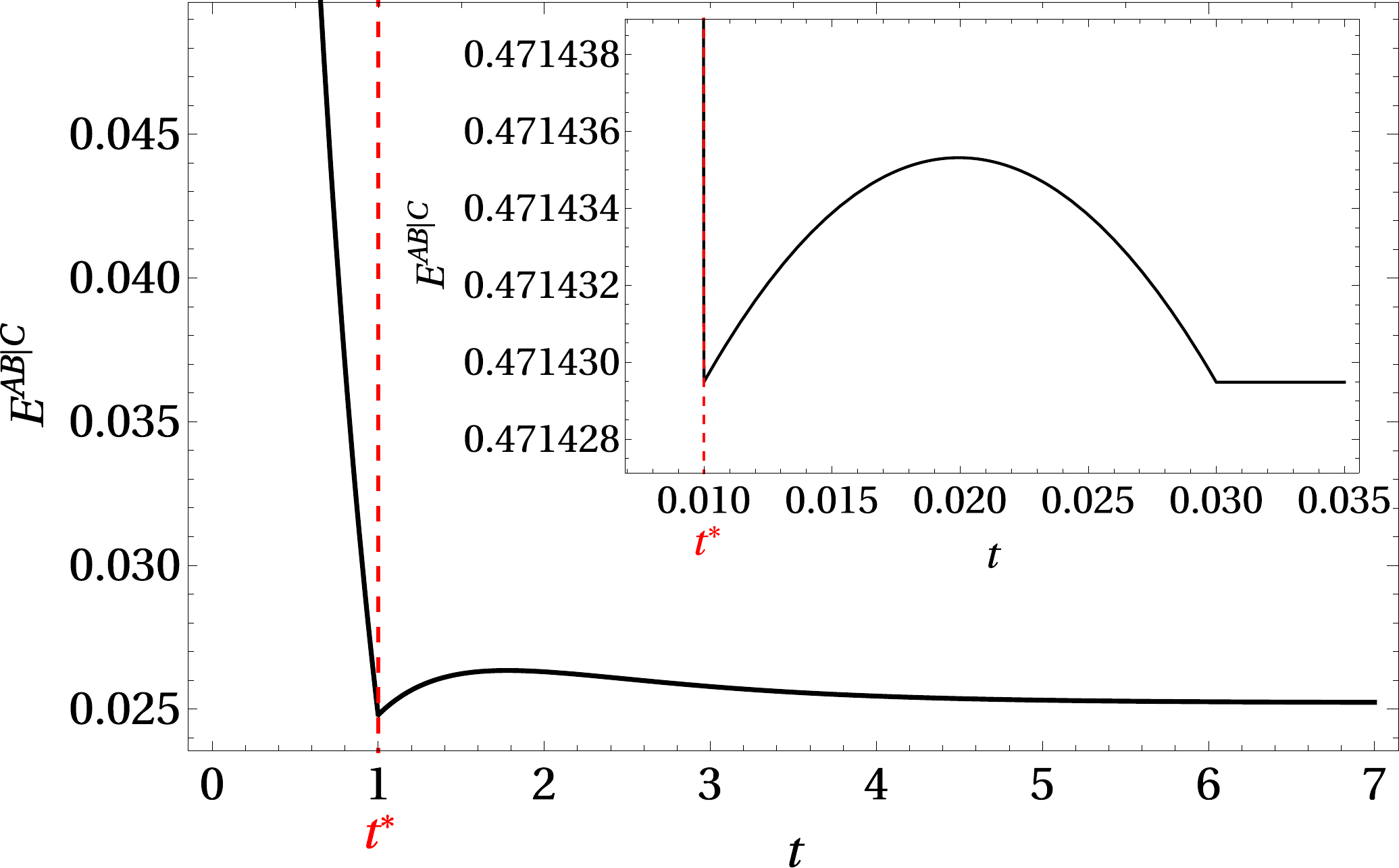}\caption{\label{fig:neg_NMwitmess_ENM} Negativity$E^{AB|C}$ as a function
of time $t$ for the eternally non-Markovian (ENM) qubit dynamics
\eqref{eq:ENM_model} with $\alpha=2$ and $c=1/2$. The initial state
$\rho^{ABC}$ has been set as in Eq.~\eqref{eq:InitialState} with
probabilities $p_{i}$ and states $\rho_{i}^{AB_{1}}$ chosen according
to the constructive method of \citet{Bylicka+2.PRL.2017}, leading
to violation of Eq.~\eqref{eq:CP-div_by_tr-norm} for a specific
time $t^{*}>0$. The plot shows detection of non-Markovianity at $t^{*}=1$
($t^{*}=0.01$ in the inset), which is marked on the axis and with
a dashed red vertical line.}
\end{figure}

However, we explicitly demonstrate that, in accordance with the Theorem~\ref{thm:neg_is_awesome},
negativity in the tripartite setting, $E^{AB|C}$, can be used to
faithfully witness the non-Markovianity of the ENM evolution for any
$t^{*}>0$. In order to choose the initial state $\rho^{ABC}$ in
Eq.~\eqref{eq:InitialState}---in particular, its constituents $p_{\ell}\rho_{\ell}^{AB_{1}}$
($\ell=1,2$) such that $E^{AB|C}$ increases at a given $t^{*}>0$---we
follow the constructive method of \citet{Bylicka+2.PRL.2017}. We
choose $\rho_{\ell}^{AB_{1}}\in\mathcal{B}(\mathbb{C}_{2}\otimes\mathbb{C}_{3})$
and mixing probabilities $p_{\ell}$ such that the trace norm in Eq.~\eqref{eq:neg_witness}
is assured to increase at time $t^{*}$ \citep{Bylicka+2.PRL.2017}.
The construction with the analytic proof can be found in the Supplemental
Material. Yet, in Fig.~\ref{fig:neg_NMwitmess_ENM}, we plot the
dynamical behaviour of $E^{AB|C}$ for the ENM model \eqref{eq:ENM_model}
with $\alpha=2$ and $c=\frac{1}{2}$ after setting $\rho^{ABC}$,
so that the non-Markovianity of dynamics can be clearly witnessed
at time $t^{*}=1$ (and $t^{*}=0.01$ within the inset).

\medskip{}

\textbf{\emph{Conclusions. }}In this Letter we discuss possibilities
and limitations to detect non-Markovianity in qubit systems and beyond.
It is shown that a very general class of quantities based on contractive
functions fails to detect non-Markovianity of all qubit evolutions.
This includes widely studied quantifiers such as trace distance, fidelity,
and quantum relative entropy. It is shown that all of them fail to
witness non-Markovianity in a certain class of evolutions, which includes
eternal non-Markovian dynamics exhibiting non-Markovianity at all
times $t>0$.

If entangled systems are employed to witness non-Markovianity, we
show that the situation strongly depends on the number of particles
used. Surprisingly, for three particles $A$, $B$, and $C$ it is
possible to witness non-Markovianity of all invertible dynamics of
system $A$ by considering entanglement in the cut $AB|C$. We show
this explicitly for entanglement negativity, a computable measure
of entanglement, which is non-monotonic for any non-Markovian invertible
dynamics and a suitably chosen initial state. For single-qubit evolutions
our results apply also when the dynamics is not invertible. As an
example, we show results for the eternal non-Markovianity model, where
the non-monotonic behavior of negativity can be observed at arbitrary
small times.

Our results demonstrate that well-established entanglement quantifiers
can be useful as faithful non-Markovianity witnesses for very general
classes of evolutions. An important question left open in this work
is whether entanglement measures can universally witness non-Markovianity
of all evolutions, incuding non-invertible dynamics beyond qubits.
Recalling that entanglement theory is a prominent example of more
general quantum resource theories, the fundamental connection between
entanglement and non-Makovianity presented in our work can also be
useful for the development of a resource theory of non-Markovianity~\citep{Bhattacharya1803.06881,Anand1903.03880}.

This work was supported by the ''Quantum Optical Technologies''
project, carried out within the International Research Agendas programme
of the Foundation for Polish Science co-financed by the European Union
under the European Regional Development Fund.

\bibliographystyle{apsrev4-1}
\bibliography{NegBib}

\appendix

\section*{Supplemental Material}

\section{Contractive functions under positive maps}

Let $f(\rho,\sigma)$ be a function which is contractive under positive
trace-preserving maps $P$, i.e., 
\begin{equation}
f\left(\rho,\sigma\right)\geq f\left(P[\rho],P[\sigma]\right).\label{eq:P-div}
\end{equation}
We will now show that any such function fulfills 
\begin{equation}
\frac{\mathrm{d}}{\mathrm{d}t}f\left(\Lambda_{t}\left[\rho\right],\Lambda_{t}\left[\sigma\right]\right)\leq0
\end{equation}
for any P-divisible evolution $\Lambda_{t}$. For this, it is enough
to note that
\begin{align}
f\left(\Lambda_{t}[\rho],\Lambda_{t}[\sigma]\right) & =f\left(V_{t,s}\circ\Lambda_{s}[\rho],V_{t,s}\circ\Lambda_{s}[\sigma]\right)\\
 & \leq f\left(\Lambda_{s}[\rho],\Lambda_{s}[\sigma]\right),\nonumber 
\end{align}
for any $0\leq s\leq t$, where we used the fact that $V_{t,s}$ is
a positive trace-preserving map for any P-divisible evolution $\Lambda_{t}$.

\section{\label{sec:Monotonic-functions}Monotonically decreasing functions
and entanglement measures cannot witness non-Markovianity}

Here we will show that a contractive function $f$ cannot witness
non-Markovianity of $\Lambda_{t}$ if $f(\Lambda_{t}[\rho],\Lambda_{t}[\sigma])$
is monotonically decreasing. We will show this for the case of discrete
time steps $t_{i}$ with $t_{0}=0$. Then, there exists a CP-divisible
family of maps $W_{t,s}$ such that 
\begin{equation}
f\left(\mu_{i},\tau_{i}\right)=f\left(\Lambda_{t_{i}}[\rho],\Lambda_{t_{i}}[\sigma]\right)\label{eq:Profile}
\end{equation}
is true for all $i$, where the states $\mu_{i}$ and $\tau_{i}$
are defined recursively via 
\begin{equation}
\mu_{i+1}=W_{t_{i+1},t_{i}}[\mu_{i}],\,\,\,\,\,\tau_{i+1}=W_{t_{i+1},t_{i}}[\tau_{i}],
\end{equation}
and $\mu_{0}=\rho$, $\tau_{0}=\sigma$. CP-divisible family $W_{t,s}$
which achieves this is given by 
\begin{equation}
W_{t_{i+1},t_{i}}[\rho]=a_{i}^{t_{i+1}-t_{i}}\rho+(1-a_{i}^{t_{i+1}-t_{i}})\frac{\openone_{d}}{d},\label{eq:CPsimulation}
\end{equation}
where the parameters $0\leq a_{i}\leq1$ are chosen such that Eq.~(\ref{eq:Profile})
is fulfilled. By continuity, such values for $a_{i}$ always exist,
as $f(W_{t_{i+1},t_{i}}[\rho],W_{t_{i+1},t_{i}}[\sigma])$ monotonically
decreases with decreasing $a_{i}$, achieving minimal value for $a_{i}=0$.

By similar arguments it follows that the behavior of any entanglement
measure $E^{AB|C}$ cannot witness non-Markovianity of $\Lambda_{t}$
if $E^{AB|C}(\Lambda_{t}^{A}\otimes\openone^{BC}[\rho^{ABC}])$ is
monotonically decreasing with $t$. For this, we recursively define
tripartite states 
\[
\mu_{i+1}^{ABC}=W_{t_{i+1},t_{i}}^{A}\otimes\openone^{BC}[\mu_{i}^{ABC}]
\]
with $\mu_{0}^{ABC}=\rho^{ABC}$, and $W_{t,s}$ is a local CP-divisible
family defined in Eq.~\eqref{eq:CPsimulation}. Here, the parameters
$a_{i}$ are chosen such that 
\begin{equation}
E^{AB|C}\left(\mu_{i}^{ABC}\right)=E^{AB|C}\left(\Lambda_{t_{i}}^{A}\otimes\openone^{BC}[\rho^{ABC}]\right).
\end{equation}
Again, such values of $a_{i}$ always exist by continuity, as $E^{AB|C}(W_{t_{i+1},t_{i}}^{A}\otimes\openone^{BC}[\rho])$
monotonically decreases with decreasing $a_{i}$, achieving minimal
value for $a_{i}=0.$

\section{Indecomposable positive maps}

Positive linear maps which admit the decomposition 
\begin{equation}
P[\rho]=p\EE_{1}[\rho]+(1-p)\EE_{2}[\rho^{T}]\label{Eq:Our.P.Map}
\end{equation}
are a subset of \emph{decomposable maps}, that is those positive maps
which can be decomposed as a sum of a CP and co-CP map: 
\begin{equation}
P_{\mathrm{dec}}[\rho]=\tilde{\mathcal{E}}_{1}[\rho]+\tilde{\mathcal{E}}_{2}[\rho^{T}],\label{Def:Decomposable.Map}
\end{equation}
where $\tilde{\mathcal{E}}_{i}$ are (not necessarily trace-preserving)
CP maps. An example of a trace-preserving positive map which cannot
be decomposed as \eqref{Eq:Our.P.Map} is the following: $P\colon M_{3}\to M_{3}$
($M_{n}$ are $n$-by-$n$ matrices over complex numbers, $n>1$)
given by 
\begin{equation}
P\left[\begin{pmatrix}a_{11} & a_{12} & a_{13}\\
a_{21} & a_{22} & a_{23}\\
a_{31} & a_{32} & a_{33}
\end{pmatrix}\right]=\frac{1}{3}\begin{pmatrix}a_{11}+2a_{22} & -a_{12} & -a_{13}\\
-a_{21} & a_{22}+2a_{33} & -a_{23}\\
-a_{31} & -a_{32} & a_{33}+2a_{11}
\end{pmatrix}.\label{Ex:Choi.Map}
\end{equation}
\citet{Choi1975285} has showed that $P$ cannot be decomposed as~\eqref{Def:Decomposable.Map}
and hence neither as \eqref{Eq:Our.P.Map}.

For $m,n\geq2$, all positive maps $P\colon M_{m}\to M_{n}$ are decomposable
for $m+n\leq5$ and for all $m+n>5$ there are indecomposable positive
maps~\citep{Stormer2013}. The example in \eqref{Ex:Choi.Map} was
the first indecomposable map, given by \citet{Choi1975285}, for $m=n=3$,
and \citet{WORONOWICZ1976165} gave the first indecomposable map for
$m=2$, $n=4$.

\section{\label{sec:Negativity}Negativity and local positive maps}

Here we will show that negativity is monotonic under local positive
maps of the form 
\begin{equation}
P^{A}\otimes\openone^{B}[\rho]=p\EE_{1}^{A}\otimes\openone^{B}[\rho]+(1-p)\EE_{2}^{A}\otimes\openone^{B}\left[\rho^{T_{A}}\right],
\end{equation}
for any CPTP maps $\EE_{i}$, bipartite state $\rho=\rho^{AB}$, and
probability $p$. Noting that $P^{A}$ commutes with partial transposition
$T_{B}$, we obtain 
\begin{align}
 & E^{A|B}\left(P^{A}\otimes\openone^{B}\left[\rho\right]\right)=\frac{1}{2}\left(\left\Vert P^{A}\otimes\openone^{B}\left[\rho^{T_{B}}\right]\right\Vert _{1}-1\right)\\
 & \leq\frac{1}{2}\left(p\left\Vert \EE_{1}^{A}\otimes\openone^{B}\left[\rho^{T_{B}}\right]\right\Vert _{1}+(1-p)\left\Vert \EE_{2}^{A}\otimes\openone^{B}\left[\rho^{T_{AB}}\right]\right\Vert _{1}-1\right)\nonumber \\
 & =\frac{p}{2}\left(\left\Vert \EE_{1}^{A}\otimes\openone^{B}\left[\rho^{T_{B}}\right]\right\Vert _{1}-1\right)\nonumber \\
 & =pE^{A|B}\left(\EE_{1}^{A}\otimes\openone^{B}\left[\rho\right]\right)\leq E^{A|B}(\rho),\nonumber 
\end{align}
where we used convexity of the trace norm and its monotonicity under
CPTP maps, and the fact that $||\EE_{2}^{A}\otimes\openone^{B}[\rho^{T_{AB}}]||_{1}=1$.

\section{No bound entanglement for states in Eq.~(16)}

Here we will show that states defined in Eq.~(16) of the main text
are never bound entangled in the bipartition $AB|C$. For this, will
show below that all states defined in Eq.~(16) fulfill the inequality
\begin{equation}
S(\rho^{AB})\geq S(\rho^{ABC}),\label{eq:BE}
\end{equation}
and that they are separable if $S(\rho^{AB})=S(\rho^{ABC})$. Noting
that a sufficient criterion for distillability of a general state
$\rho^{ABC}$ in the bipartition $AB|C$ is that $S(\rho^{AB})>S(\rho^{ABC})$~\citep{DevetakRspa.2004.1372},
this proves that none of the states defined in Eq.~(16) is bound
entangled.

To show that the inequality~(\ref{eq:BE}) is fulfilled by all states
in Eq.~(16), note that 
\begin{align}
S(\rho^{ABC}) & =h(p_{1})+p_{1}S(\rho_{1}^{AB_{1}})+p_{2}S(\rho_{2}^{AB_{1}}),\\
S(\rho^{AB}) & =1+S(p_{1}\rho_{1}^{AB_{1}}+p_{2}\rho_{2}^{AB_{1}}).
\end{align}
where $h(x)=-x\log_{2}x-(1-x)\log_{2}(1-x)$ is the binary entropy.
Using concavity of the von Neumann entropy and the fact that $h(p_{1})\leq1$,
we obtain the following: 
\begin{align}
S(\rho^{AB}) & =1+S(p_{1}\rho_{1}^{AB_{1}}+p_{2}\rho_{2}^{AB_{1}})\label{eq:BE-2}\\
 & \geq1+p_{1}S(\rho_{1}^{AB_{1}})+p_{2}S(\rho_{2}^{AB_{1}})\nonumber \\
 & \geq h(p_{1})+p_{1}S(\rho_{1}^{AB_{1}})+p_{2}S(\rho_{2}^{AB_{1}})=S(\rho^{ABC}),\nonumber 
\end{align}
which proves Eq.~(\ref{eq:BE}). 

In case that $S(\rho^{AB})=S(\rho^{ABC})$ both inequalities in Eq.~(\ref{eq:BE-2})
must hold with equality, which implies that 
\begin{equation}
\rho_{1}^{AB_{1}}=\rho_{2}^{AB_{1}},\,\,\,\,\,\,p_{1}=p_{2}=\frac{1}{2}.
\end{equation}
It is straightforward to verify that in this case the state in Eq.~(16)
is separable in the bipartition $AB|C$.

\section{\label{sec:ENM}Eternally non-Markovian qubit dynamics}

For the general solution to the master equation (21) of the main text
describing random unitary dynamics we refer the reader to Ref.~\citet{Chruscinski2013}.
Still, for the choice of decay parameters (23) corresponding to the
ENM model, the mixing probabilities $p_{\mu}(t)$ in the dynamical
(Pauli) map $\Lambda_{t}$ defined in (22) read: \begin{subequations}
\begin{align}
p_{0}^{(\alpha)}(t) & =\frac{1}{4}\left[1+e^{-2\alpha ct}\left(1+2e^{\alpha ct}\cosh^{\alpha}(ct)\right)\right],\\
p_{1}^{(\alpha)}(t) & =p_{2}^{(\alpha)}(t)=\frac{1}{4}\left(1-e^{-2\alpha ct}\right),\\
p_{3}^{(\alpha)}(t) & =\frac{1}{4}\left[1+e^{-2\alpha ct}\left(1-2e^{\alpha ct}\cosh^{\alpha}(ct)\right)\right],
\end{align}
\end{subequations}where with a superscript we have specially stated
the dependence on the parameter $\alpha\ge1$. We have kept the $\alpha$-dependence
explicit, so that we can conveniently express the inverse map of $\Lambda_{t}$
(i.e., $\Lambda_{t}^{-1}$~s.t.~$\Lambda_{t}^{-1}\circ\Lambda_{t}=\openone$),
which also takes the Pauli form (22), as $\Lambda_{t}^{-1}\left[\rho\right]=\sum_{\mu}p_{\mu}^{(-\alpha)}(t)\sigma_{\mu}\rho\sigma_{\mu}$
by simply changing the sign of $\alpha$.

As a result, the CJ matrix of the map $V_{t,s}=\Lambda_{t}\circ\Lambda_{s}^{-1}$
stated in Eq.~(24) can be directly computed as 
\begin{align}
\Omega_{V_{t,s}} & =\left(\Lambda_{t}\circ\Lambda_{s}^{-1}\right)\otimes\openone\left[P_{\Phi_{+}}\right]\\
 & =\sum_{\mu,\nu}p_{\mu}^{(\alpha)}(t)\,p_{\nu}^{(-\alpha)}(s)\;\sigma_{\mu}^{A}\sigma_{\nu}^{A}P_{\Phi_{+}}\sigma_{\nu}^{A}\sigma_{\mu}^{A}\nonumber \\
 & =\left(\sum_{\mu}p_{\mu}^{(\alpha)}(t)\,p_{\mu}^{(-\alpha)}(s)\right)\;P_{\Phi_{+}}\nonumber \\
 & +\sum_{i\ne j}p_{i}^{(\alpha)}(t)\,p_{j}^{(-\alpha)}(s)\;\sigma_{i}^{A}\sigma_{j}^{A}P_{\Phi_{+}}\sigma_{j}^{A}\sigma_{i}^{A}\nonumber \\
 & +\sum_{i}\left(p_{0}^{(\alpha)}(t)\,p_{i}^{(-\alpha)}(s)+p_{i}^{(\alpha)}(t)\,p_{0}^{(-\alpha)}(s)\right)\;\sigma_{i}^{A}P_{\Phi_{+}}\sigma_{i}^{A},\nonumber 
\end{align}
where $P_{\psi}=2\ket{\psi}\!\bra{\psi}$. Using, the properties of
Pauli operators and Bell states (e.g., $\sigma_{1}^{A}\left|\Phi_{+}\right\rangle =\left|\Psi_{+}\right\rangle $),
as well as $P_{\Psi_{+}}^{T_{B}}=\openone_{4}-P_{\Phi_{-}}$, one
arrives at the decomposition (25): 
\begin{equation}
\Omega_{V_{t,s}}=p_{1}P_{\Phi_{+}}+p_{2}P_{\Phi_{-}}+(1-p_{1}-p_{2})P_{\Psi_{+}}^{T_{B}},
\end{equation}
with \begin{subequations}
\begin{align}
p_{1} & =\frac{1}{2}e^{-c\,(t-s)\,\alpha}\left[e^{-c\,(t-s)\,\alpha}+\left(\cosh(cs)\,\text{sech}(ct)\right)^{-\alpha}\right],\\
p_{2} & =\frac{1}{2}\left[1-e^{-c\,(t-s)\,\alpha}\left(\cosh(cs)\,\text{sech}(ct)\right)^{-\alpha}\right].
\end{align}
\end{subequations}In order to prove that $\{p_{1},p_{2},1-p_{1}-p_{2}\}$
constitutes a valid probability distribution, it is enough to demonstrate
that $p_{1},p_{2}\geq0$ and $p_{1}+p_{2}\leq1$. Being a sum of nonnegative
quantities, clearly $p_{1}\ge0$. Since $s\leq t$ and $c,\,\alpha$
are positive, 
\begin{equation}
0\leq e^{-2c\,(t-s)\,\alpha}\leq1,
\end{equation}
showing that $p_{1}+p_{2}\leq1$. Thus, it remains only to show that
$p_{2}\geq0$, which is equivalent to 
\begin{align}
 & \left[e^{c\,(t-s)}\cosh(cs)\,\text{sech}(ct)\right]^{-\alpha}\leq1\\
\Leftrightarrow\quad & e^{c\,(t-s)}\cosh(cs)\,\text{sech}(ct)\geq1\nonumber \\
\Leftrightarrow\quad & \frac{1+e^{-2c\,s}}{1+e^{-2c\,t}}\geq1\nonumber \\
\Leftrightarrow\quad & s\leq t,\nonumber 
\end{align}
which is true.

\section{\label{sec:P_decomp}Choi-Jamio\l kowski matrix decomposition for
the P-maps of interest}

The action of any linear TP-map $\Lambda$ on $\rho\in\mathcal{B}(\mathcal{H}_{d})$
can be generally expressed as 
\begin{equation}
\Lambda\left[\rho\right]=\mathrm{Tr}_{B}\!\left\{ \Omega_{\Lambda}(\openone_{d}\otimes\rho^{T})\right\} ,
\end{equation}
where $\Omega_{\Lambda}:=\Lambda\otimes\openone\left[d\left|\Phi^{+}\right\rangle \!\left\langle \Phi^{+}\right|\right]$
is the ``effective'' CJ-matrix satisfying $\mathrm{Tr}_{A}\Omega_{\Lambda}=\openone_{d}$,
yet not being necessarily positive semi-definite.

Now, let us show that if $\Omega_{\Lambda}$ admits a convex decomposition:
\begin{equation}
\Omega_{\Lambda}=p\,\Omega_{\mathcal{E}_{1}}+(1-p)\,\Omega_{\mathcal{E}_{2}}^{T_{B}}
\end{equation}
with $0\le p\le1$, $\Omega_{\mathcal{E}_{\ell}}\ge0$ and $\mathrm{Tr}_{A}\Omega_{\mathcal{E}_{\ell}}=\openone_{d}$
for both $\ell=1,2$, then the linear map $\Lambda$ can always be
decomposed as stated in the main text in Eq.~(5).

This is because one may then explicitly write: 
\begin{align}
\Lambda[\rho] & =\mathrm{Tr}_{B}\!\left\{ \Omega_{\Lambda}\,(\openone\otimes\rho^{T})\right\} \\
 & =\mathrm{Tr}_{B}\!\left\{ \left(p\,\Omega_{\mathcal{E}_{1}}+(1-p)\,\Omega_{\mathcal{E}_{2}}^{T_{B}}\right)(\openone\otimes\rho^{T})\right\} \nonumber \\
 & =p\,\mathrm{Tr}_{B}\!\left\{ \Omega_{\mathcal{E}_{1}}(\openone\otimes\rho^{T})\right\} +(1-p)\,\mathrm{Tr}_{B}\!\left\{ \left(\Omega_{\mathcal{E}_{2}}\right)(\openone\otimes\rho)\right\} \nonumber \\
 & =p\,\mathcal{E}_{1}[\rho]+(1-p)\,\mathcal{E}_{2}[\rho^{T}],\nonumber 
\end{align}
where $\mathcal{E}_{\ell}\left[\rho\right]=\mathrm{Tr}_{B}\!\left\{ \Omega_{\mathcal{E}_{\ell}}(\openone_{d}\otimes\rho^{T})\right\} $
are the CPTP maps defined by the (positive semi-definite) CJ-matrices
$\Omega_{\mathcal{E}_{\ell}}$.

\medskip{}

\section{\label{sec:Bogna_procedure} Constructing $\rho^{ABC}$ such that
$E^{AB|C}$ is a faithful non-Markovianity witness for a given time
instance $t^{*}>0$}

We follow the method of \citet{Bylicka+2.PRL.2017} which describes
how to construct initial states $\rho_{1}$ and $\rho_{2}$, such
that for any family of invertible dynamical maps, $\Lambda_{t}:\mathcal{B}(\mathcal{H}_{d})\to\mathcal{B}(\mathcal{H}_{d})$,
the trace distance (2) is always increasing at a given time $t^{*}>0$,
i.e., the right derivative 
\begin{align}
\left.\frac{\mathrm{d}}{\mathrm{d}t}\big\|\,\Lambda_{t}\otimes\openone_{d+1}[\rho_{1}-\rho_{2}]\,\big\|_{1}\right|_{t=t_{+}^{*}} & >0\label{eq:trace_grows}
\end{align}
defined via $\left.\frac{\mathrm{d}}{\mathrm{d}t}\bullet\right|_{t_{+}^{*}}:=\lim_{\delta t\to0_{+}}\left.\frac{\mathrm{d}}{\mathrm{d}t}\bullet\right|_{t=t^{*}+\delta t}$
is positive; whenever $V_{t^{*}+\delta t,t^{*}}=\Lambda_{t^{*}+\delta t}\circ\Lambda_{t^{*}}^{-1}$
is not CP as $\delta t\to0_{+}$, i.e., the dynamical family is not
CP-divisible at $t^{*}$. Note that thanks to considering the ancilla
above to be of dimension $d+1$, the probabilities in Eq.~(2) of
the main text can be assumed $p_{1}=p_{2}=\frac{1}{2}$ without loss
of generality.

Once $\rho_{1}$ and $\rho_{2}$ are determined, by setting the initial
tripartite state introduced in Eqs.~(16-17) as
\begin{equation}
\rho^{ABC}=\frac{1}{2}\left(\rho_{1}\otimes\ket{\Psi^{+}}\!\bra{\Psi^{+}}+\rho_{2}\otimes\ket{\Psi^{-}}\!\bra{\Psi^{-}}\right)
\end{equation}
with $\rho^{ABC}\in\mathcal{B}(\mathcal{H}_{d}^{A}\otimes\mathcal{H}_{d+1}^{B_{1}}\otimes\mathcal{H}_{2}^{B_{2}}\otimes\mathcal{H}_{2}^{C})$,
the condition \eqref{eq:trace_grows} assures the corresponding negativity
to fulfill
\begin{equation}
\left.\frac{\mathrm{d}}{\mathrm{d}t}E^{AB|C}\left(\tau_{t}^{ABC}\right)\right|_{t=t_{+}^{*}}=\frac{1}{2}\left.\frac{\mathrm{d}}{\mathrm{d}t}\big\|\,\Lambda_{t}\otimes\openone_{d+1}[\rho_{1}-\rho_{2}]\,\big\|_{1}\right|_{t=t_{+}^{*}}>0,\label{eq:negativity_grows}
\end{equation}
so that $E^{AB|C}$ can be, indeed, considered a faithful witness
of non-Markovianity for any $t^{*}>0$ at which the CP-divisibility
property of dynamics is violated.

\subsubsection*{Constructing necessary $\rho_{1}$ and $\rho_{2}$ in case of the
eternally non-Markovian qubit dynamics for a given $t^{*}>0$}

Here, we describe in %
detail the above procedure for the case of ENM qubit dynamics%
. Note that for any qubit dynamics the above construction requires
$\rho_{1},\rho_{2}\in\mathcal{B}(\mathcal{H}_{2}^{A}\otimes\mathcal{H}_{3}^{B_{1}})$,
i.e., to deal with qubit-qutrit states. The form of the dynamical
map $\Lambda_{t}$ at each $t\ge0$, as well as its inverse $\Lambda_{t}^{-1}$,
for the ENM model are described above in Sec.~\ref{sec:ENM}. Following
the method of \citet{Bylicka+2.PRL.2017}:
\begin{enumerate}
\item We choose the maximally mixed state, $\sigma=\frac{1}{6}\openone_{6}$,
as an example of a state that lies in the image of $\Lambda_{t}\otimes\openone_{3}$
for any $t\ge0$ in case of the ENM model.
\item We set $\rho^{A}=\left|0\right\rangle \!\left\langle 0\right|$ as
an exemplary state in $\mathcal{B}(\mathcal{H}_{2}^{A})$.
\item We compute states $\rho_{1}^{\prime}(\lambda)=(1-\lambda)\,\sigma+\lambda\ket{\Phi^{+}}\!\bra{\Phi^{+}}$
and $\rho_{2}^{\prime}(\lambda)=(1-\lambda)\,\sigma+\lambda\,\rho^{A}\otimes\left|2\right\rangle \left\langle 2\right|$
with $\left|\Phi^{+}\right\rangle =\frac{1}{\sqrt{2}}\left|00\right\rangle +\left|11\right\rangle $,
such that $\big\|\rho_{1}^{\prime}(\lambda)-\rho_{2}^{\prime}(\lambda)\big\|_{1}=2\lambda$.
\item For a given fixed $t^{*}>0$, we find maximal $0<\lambda\le1$ such
that both $\Lambda_{t^{*}}^{-1}\otimes\openone_{3}\!\left[\rho_{1}^{\prime}(\lambda)\right]\ge0$
and $\Lambda_{t^{*}}^{-1}\otimes\openone_{3}\!\left[\rho_{2}^{\prime}(\lambda)\right]\ge0$
are legitimate quantum states. For the ENM dynamics and above choices,
we obtain:
\begin{equation}
\lambda^{*}=\frac{1}{3e^{2\alpha ct^{*}}-2}.\label{eq:lambda_star}
\end{equation}
\item In this way, we arrive at the desired initial states that read:
\begin{align}
\rho_{1} & :=\Lambda_{t^{*}}^{-1}\otimes\openone_{3}\!\left[\rho_{1}^{\prime}(\lambda^{*})\right]\\
 & =\frac{1}{6}\;\mathrm{diag}(1-\lambda^{*},1-\lambda^{*},2+4\lambda^{*},1+\lambda^{*},1+\lambda^{*},0)\nonumber 
\end{align}
and
\begin{align}
\rho_{2} & :=\Lambda_{t^{*}}^{-1}\otimes\openone_{3}\!\left[\rho_{2}^{\prime}(\lambda^{*})\right]\\
 & =\frac{1}{2}\left(\begin{array}{cccccc}
\frac{1+\lambda^{*}}{2} & 0 & 0 & 0 & \frac{(1+2\lambda^{*})}{\chi^{\alpha}} & 0\\
0 & \frac{1-\lambda^{*}}{6} & 0 & 0 & 0 & 0\\
0 & 0 & \frac{1-\lambda^{*}}{3} & 0 & 0 & 0\\
0 & 0 & 0 & \frac{1-\lambda^{*}}{6} & 0 & 0\\
\frac{(1+2\lambda^{*})}{\chi^{\alpha}} & 0 & 0 & 0 & \frac{1+\lambda^{*}}{2} & 0\\
0 & 0 & 0 & 0 & 0 & \frac{1-\lambda^{*}}{3}
\end{array}\right),\nonumber 
\end{align}
where $\chi=\frac{1}{2}\left[\left(\frac{1+2\lambda^{*}}{\lambda^{*}}\right)^{1/\alpha}+3^{1/\alpha}\right]$.
\end{enumerate}
In order to explicitly demonstrate the correctness of the above construction
for the ENM model, we compute the resulting states at any time $t\ge0$:
$\rho_{1}(t):=\Lambda_{t}\otimes\openone_{3}[\rho_{1}]$ and $\rho_{2}(t):=\Lambda_{t}\otimes\openone_{3}[\rho_{1}]$,
whose analytic expressions we skip here due to their cumbersome form.
Yet, we explicitly write the resulting trace-distance between them:
\begin{align}
 & \big\|\rho_{1}(t)-\rho_{2}(t)\big\|_{1}=2\lambda^{*}\times\label{eq:trace_dist_ENM}\\
 & \times\begin{cases}
e^{-2\alpha c(t-t^{*})} & t\le t^{*}\\
\frac{1}{4}\left[3+e^{-2\alpha c(t-t^{*})}\left(\mathcal{R}(t,t^{*})-1\right)+\mathcal{R}(t,t^{*})\right] & t^{*}<t\le t^{\uparrow}\\
1 & t>t^{\uparrow}
\end{cases},\nonumber 
\end{align}
with $t^{\uparrow}>t^{*}$ being determined by the solution to the
transcendental equation $\mathcal{R}(t^{\uparrow},t^{*})=1$, where
\begin{equation}
\mathcal{R}(t,t^{*}):=\frac{\left(\cosh(ct^{*})\,\text{sech}(ct)\right)^{-\alpha}}{\cosh(\alpha c(t-t^{*}))}.
\end{equation}
If such solution does not exist (apart from the trivial $\mathcal{R}(t^{*},t^{*})=1$),
then $t^{\uparrow}\to\infty$ and the only non-smooth behavior of
the trace-distance occurs at $t=t^{*}$, which is the crucial one
indicating the non-Markovianity of the evolution.

In particular, after computing the right derivative of Eq.~\eqref{eq:trace_dist_ENM}
at $t\to t_{+}^{*}$, we may explicitly evaluate Eq.~\ref{eq:negativity_grows}
for the ENM model, as follows
\begin{align}
\left.\frac{\mathrm{d}}{\mathrm{d}t}E^{AB|C}\left(\tau_{t}^{ABC}\right)\right|_{t=t_{+}^{*}} & =\lim_{\delta t\to0}\frac{1}{2}\left.\frac{\mathrm{d}}{\mathrm{d}t}\big\|\rho_{1}(t)-\rho_{2}(t)\big\|_{1}\right|_{t=t^{*}+\delta t}\nonumber \\
 & =\frac{1}{2}\alpha c\lambda^{*}\tanh(ct^{*})>0,
\end{align}
which consistently is positive for any $t^{*}>0$ (due to $\lambda^{*}>0$
in Eq.~\eqref{eq:lambda_star} and $\forall_{x>0}:\;\tanh(x)>0$).

In Fig.1 of the main text, in order to more directly show the increasing
behavior of the negativity (20) as a non-Markovianity witness, we
plot rather the full dynamical behaviour of the above $E^{AB|C}$
as a function of $t$, i.e.:
\begin{equation}
E^{AB|C}=\frac{1}{2}\big\|\rho_{1}(t)-\rho_{2}(t)\big\|_{1}=\frac{1}{2}\big\|\,\Lambda_{t}\otimes\openone_{3}[\rho_{1}-\rho_{2}]\,\big\|_{1},
\end{equation}
for particular values of the ENM parameters $\alpha\ge1$ and $c>0$.
We, however, choose different values of $t^{*}>0$ to show that non-Markovianity
can be witnessed this way at arbitrary timescales.
\end{document}